\documentclass[year=25]{fmcad}
\usepackage{xcolor}
\usepackage{xspace} 
\usepackage{amsmath}
\usepackage{amssymb}
\usepackage{url}
\usepackage{bm}
\usepackage{booktabs}
\usepackage{simplebnf}
\newcommand{\compr}[2]{[#1 \mid #2]}

\newcommand{\sleec}{SLEEC\xspace}

\newcommand{\str}[1]{\mathcal{#1}}
\newcommand{\conf}[1]{( #1 )}

\newcommand{\rem}[1]{}
\renewcommand{\rem}[1]{\textcolor{red}{[#1]}}

\newcommand{\mm}[1]{}
\renewcommand{\mm}[1]{[\textcolor{blue}{MM: #1}]}

\newcommand{\teq}{\approx}
\newcommand{\tneq}{\not\approx}

\newcommand{\ter}[1]{\mathcal{T}(#1)}

\newcommand{\sbool}{\mathsf{Bool}}

\newcommand{\ite}{\mathsf{ite}}

\newcommand{\ltrue}{\top} 

\newcommand{\Mo}{\mathbf{I}}

\newcommand{\Sc}{\mathsf{S}}

\newcommand{\Ec}{\mathsf{E}}

\newcommand{\card}[1]{\lvert #1 \rvert}

\newcommand{\define}[1]{\textsl{#1}}

\newcommand{\unknown}{\ensuremath{\mathsf{unknown}}\xspace}
\newcommand{\unsat}{\ensuremath{\mathsf{unsat}}\xspace}
\newcommand{\sat}{\ensuremath{\mathsf{sat}}\xspace}
\newcommand{\timeout}{\ensuremath{\mathsf{timeout}}\xspace}

\newcommand{\cvc}{{\small cvc5}\xspace}

\newcommand{\lian}{\tname{LIA}}

\newcommand{\bool}{\ensuremath{\mathsf{Bool}}\xspace}

\newcommand{\I}{\mathcal{I}}

\renewcommand{\vec}[1]{\bar{#1}}
\newcommand{\tname}[1]{\mathsf{#1}}

\newcommand{\thsym}[1]{\ensuremath{\mathsf{#1}}\xspace}

\newcommand{\s}{\thsym{Set}}

\newcommand{\relation}{\thsym{Rel}}
\newcommand{\tuple}{\thsym{Tuple}}
\renewcommand{\int}{\thsym{Int}}

\newcommand{\map}{\pi\xspace}

\newcommand{\eleSort}{\ensuremath{\varepsilon}\xspace}

\newcommand{\sqin}{%
  \mathrel{\vphantom{\sqsubset}\text{%
      \mathsurround=0pt
      \ooalign{$\sqsubset$\cr$-$\cr}%
    }}%
}

\newcommand{\filter}{\sigma}

\newcommand{\product}{\times}

\newcommand\squplus{\mathbin{\ooalign{$\sqcup$\cr%
      \hfil\raise0.42ex\hbox{$\scriptscriptstyle +$}\hfil\cr}}}

\newif \ifFANCYOPER
\FANCYOPERfalse

\newcommand{\vars}{\mathrm{Vars}}
\newcommand{\arity}{\mathrm{ar}}
\newcommand{\tup}[1]{\langle#1\rangle}
\newcommand{\seq}[2]{\bar{#1}_{#2}}

\newcommand{\symFont}[1]{\ensuremath{\mathsf{#1}}}

\newcommand{\tupleSort}{\symFont{Tup}}

\newcommand{\sigmap}{{\mathcal{L}_{\pi}}}

\newcommand{\sigHilbert}{{\mathcal{H}}}
\newcommand{\lambdaExpr}{e^{\lambda}}
\newcommand{\arithExpr}{e^{\lian}}

\newcommand{\opemptyset}{[\,]}
\newcommand{\sempty}{[\,]}

\newcommand{\opprod}{\product}

\newcommand{\opunion}{\sqcup}
\newcommand{\opinter}{\sqcap}
\newcommand{\opsetminus}{\setminus}
\newcommand{\opin}{\mathrel{\ooalign{$\sqsubset$\cr{$-$}}}}
\newcommand{\opnotin}{\not\opin}
\newcommand{\opequal}{\approx}

\newcommand{\opsingleton}[1]{[\,#1\,]}

\newcommand{\Scclosed}{\Sc^*}

\newcommand{\termsof}[1]{\mathcal{T}(#1)}

\newcommand{\cadd}[2]{#1,\, #2}

\newcommand{\rn}[1]{\textsc{\small#1}\xspace}

\newcommand{\ruleInterDown}{\rn{Inter Down}}
\newcommand{\ruleInterUp}{\rn{Inter Up}}

\newcommand{\ruleUnionDown}{\rn{Union Down}}
\newcommand{\ruleUnionUp}{\rn{Union Up}}

\newcommand{\ruleDifferenceDown}{\rn{Diff Down}}
\newcommand{\ruleDifferenceUp}{\rn{Diff Up}}

\newcommand{\ruleProductUp}{{\rn{Prod Up}}\xspace}
\newcommand{\ruleProductDown}{{\rn{Prod Down}}\xspace}

\newcommand{\ruleSingleUp}{{\rn{Single Up}}\xspace}
\newcommand{\ruleSingleDown}{\rn{Single Down}}

\newcommand{\ruleSetDiseq}{\rn{Set Diseq}}
\newcommand{\ruleEqUnsat}{\rn{Eq Unsat}}
\newcommand{\ruleSetUnsat}{\rn{Set Unsat}}
\newcommand{\ruleEmptyUnsat}{\rn{Empty Unsat}}

\newcommand{\T}{\mathbb{T}}

\newcommand{\relationsTheory}{\T_\mathsf{Rel}}
\newcommand{\relationsSig}{\Sigma_\mathsf{Rel}}
\newcommand{\elemsTheory}{\T_\mathsf{El}}
\newcommand{\elemsSig}{\Sigma_\mathsf{El}}
\newcommand{\ver}{\; \vert \;}

\newcommand{\ruleEConf}{\rn{E-Conf}}

\newcommand{\ruleEIdent}{\rn{E-Ident}}

\newcommand{\ruleTableProductUp}{\rn{Prod Up}} 
\newcommand{\ruleTableProductDown}{\rn{Prod Down}} 

\newcommand{\ruleSetFilterUp}{\rn{Filter Up}}
\newcommand{\ruleSetFilterDown}{\rn{Filter Down}}

\newcommand{\inferR}[3][2]{\infer[#1]{#3}{#2}}

\newcommand{\terminatingSetCount}{14}

\newcommand{\filterDecidable}{{\mathcal{F}}}
\newcommand{\setAll}{\mathsf{set.all}}
\newcommand{\setSome}{\mathsf{set.some}}

\newcommand{\universals}{\(\text{MLSSF}^{\forall}\)}
\usepackage{bussproofs}
\usepackage{proof}
\usepackage{thmtools, thm-restate}
\usepackage{version}
\usepackage{etoolbox}
\usepackage{threeparttable}
\usepackage[capitalise,noabbrev]{cleveref}
\usepackage{amssymb, amsmath, amsthm}
\usepackage{thmtools, thm-restate}
\usepackage{cuted}

\newtheorem{condition}{Condition}

\newtoggle{paper}
\toggletrue{paper}      
\title{Solving Set Constraints with Comprehensions and Bounded Quantifiers}
\DeclareRobustCommand{\IEEEauthorrefmark}[1]{\textsuperscript{#1}}

\author{\IEEEauthorblockN{
  Mudathir Mohamed\orcid{0000-0002-4644-5756}\IEEEauthorrefmark{1},
  Nick Feng\orcid{0009-0002-3908-2755}\IEEEauthorrefmark{2},  
  Andrew Reynolds\orcid{0000-0002-3529-8682}\IEEEauthorrefmark{1},
  Cesare Tinelli\orcid{0000-0002-6726-775X}\IEEEauthorrefmark{1},  
  Clark Barrett\orcid{0000-0002-9522-3084}\IEEEauthorrefmark{3},
  and
  Marsha Chechik\orcid{0000-0002-6301-3517}\IEEEauthorrefmark{2}
  }
\and
\IEEEauthorblockA{
  \and
\IEEEauthorrefmark{1}\textit{The University of Iowa}
\and
\IEEEauthorrefmark{2}\textit{University of Toronto}
\and
\IEEEauthorrefmark{3}\textit{Stanford University}
}
}

\begin{document}
\maketitle
\iftoggle{paper}{
 \includeversion{paper}
 \excludeversion{report}
}{
 \excludeversion{paper}
 \includeversion{report}
}
%
%
\begin{abstract}
  Many real applications problems can be encoded easily as 
  quantified formulas in SMT.
  However, this simplicity comes at the cost of difficulty during solving
  by SMT solvers.
  Different strategies and quantifier instantiation techniques
  have been developed to tackle this.
  However, 
  SMT solvers still struggle with quantified formulas
  generated by some applications.
  In this paper, we discuss the use of set-bounded quantifiers,
  quantifiers whose variable ranges over a finite set.
  These quantifiers can be implemented using
  quantifier-free fragment of the theory of finite relations
  with a filter operator,
  a form of restricted comprehension, 
  that constructs a subset 
  from a finite set using a predicate.
  We show that this approach outperforms other quantification techniques in
  satisfiable problems generated by the \sleec tool, and is very competitive 
  on unsatisfiable 
  benchmarks compared to LEGOS, a specialized solver for \sleec.
  We also identify a decidable class of constraints with restricted
  applications of the filter operator,
  while showing that unrestricted applications lead to undecidability.
  %
\end{abstract}
%
%
\section{Introduction}
%
%
Problems from many real applications can be encoded naturally in the language of finite sets and relations.
Concrete examples include software design specifications~\cite{alloyRef},
ontologies~\cite{relationsPaper}, database queries~\cite{lparBags}, 
normative requirements~\cite{sleec}, and authorization policies~\cite{cedarPaper}.
Some SMT solvers support the theory of finite sets and relations~\cite{sets,relationsPaper}.
Possible encodings in SMT from the applications listed above
involve the use of quantifiers with variables that range over elements 
of some sets with possibly unbounded but finite cardinality. 
A typical approach is to use standard quantifiers 
and rely on general quantifier instantiation techniques implemented in SMT solvers. 
However, solver performance is generally poor in these cases
because of the inherent difficulty of reasoning about quantified formulas.

The \cvc solver supports a theory of finite sets and relations that was recently
extended with a second-order filter operator~\cite{lparBags}.
The filter operator implements a form of set restricted comprehension 
that constructs a subset from a given set using a given predicate.
The extension makes it now possible to express set-bounded quantification
using only quantifier-free formulas.
This presents an opportunity to encode problems that involve
bounded quantification as quantifier-free problem,
avoiding the challenges of reasoning with quantifiers.

In this paper, we discuss how to encode set-bounded quantifiers
in a subtheory of \cvc's theory of finite sets and relations,
and present very encouraging initial experimental results that evaluate
this encoding on a set of benchmarks from real-world problems. 
%
%
We also discuss the theoretical question on whether there is 
a reasonably expressive logical fragment with set-bounded quantification 
with a decidable satisfiability problem.
We answer the question positively, but also show that restrictions
on the use of the filter operator are not just sufficient but necessary
for decidability.

%
%
%

%
%
This work makes the following \textbf{contributions}:
\begin{enumerate}
\item 
A decision procedure for the satisfiability 
of a large class of quantifier-free formulas
for a theory of finite sets and relations extended 
with the Cartesian product operator and a filter operator.
\item 
A high-level proof that the unrestricted use of
the filter operator compromises our decidability result.
\item
An approach for expressing set-bounded quantifiers
in terms of quantifier-free formulas containing applications of the filter operator.
\item
Initial empirical results over a set of benchmarks from a real-world application 
showing that an encoding to SMT using our approach
leads to better solver performance compared to encodings using standard quantifiers.
\end{enumerate}

%
%
\subsection{Related Work}
Multi-Level Syllogistic (MLS) theory can be taken as the core fragment
for unsorted set theory~\cite{setMapVariablePaper}.
It contains the basic operators \(\sqcup, \sqcap, \setminus, \sqsubseteq, \teq ,\sqin\)
in addition to logical connectives \(\vee, \wedge, \neg\).
Its syntax and semantics consider only sets, with no scalar values,
as the latter can be encoded as (nested) sets.
A decision procedure based on singleton models\footnote{
  A singleton model is one that assigns each set variable
  either \(\varnothing\) or \(\{\varnothing\}\).} was given
by Ferro et al.~\cite{setMapVariablePaper}.
The same paper proved the decidability of an MLS extension with
map variables\footnote{Map variables are binary relations.}
and operators that return as sets their domains and ranges.
%
%

Bansal et al.~\cite{sets} presented
an efficient decision procedure for a natural variant of MLS with a set cardinality operator in the context of many-sorted logic. 
Meng et al.~\cite{relationsPaper} 
extended that theory with operators for finite relations, expressed as finite sets of tuples, 
that include cross product, relational join, transpose, 
and transitive closure. 
They proved the decidability of a small restricted fragment that only accepts  
binary relations along with relational join and transpose operators. 
Recently, Mohamed et al.~\cite{lparBags} extended that theory further  
with filter \(\filter\) and map \(\map\) operators 
to support reasoning about SQL queries.
%
%
It can be shown that 
MLS extended with a set map operator \(\map\) is decidable
when its function argument is uninterpreted.
This can be achieved by a reduction to the MLS extension with
map variables in~\cite{setMapVariablePaper}.
%
%
If interpreted functions are allowed in map terms, 
the fragment is undecidable in general, similarly to the undecidability result 
we present in~\cref{sec:undecidable_set_fragments}.

Set-bounded quantifiers have been studied extensively for 
unsorted theories of finite sets. 
Cantone et al.~\cite{quantifiersUndecidable} showed the undecidability 
of alternating set quantifiers \((\forall \exists)_0\).
Some restricted fragments that are decidable are mentioned in
Cantone~\cite{survey}. 
\universals is an example of a decidable fragment 
that extends MLS with uninterpreted unary functions, 
singletons, and positive universal bounded quantifiers only
with some restrictions on their body expressions. 
\universals is close in spirit to the decidable fragment we present
in~\cref{sec:decidable} for our many-sorted theory of sets, 
albeit the latter allows existential quantifiers in restricted cases. 
Cristi\'{a} et al.~\cite{prolog} discussed decidable fragments for ``restricted'' quantifiers 
(our ``set-bounded'' quantifiers): 
\(\exists, \forall, \exists \forall\), and \(\forall \exists\) 
with no cycles in the ``domain graph''.
Our calculus and implementation in \cvc
decides the satisfiability of formulas in the first three fragments 
\(\exists, \forall, \exists \forall\) but not in \(\forall \exists\).
Currently, it accepts \(\forall \exists\) formulas but without termination guarantees. 
That said, our experimental evaluation does include \(\forall \exists\) problems with
alternating quantifiers, 
and our approach works pretty well with those problems. 
We also note that our calculus supports the cross product operator, 
which is not mentioned in~\cite{prolog}.
%

Deciding MLS extended with Cartesian product is reported as an open problem
by 
\cite{undecidable_hilbert_set,survey}.
Cantone and Ursino~\cite{cantoneDecidableProduct} proved that MLS with unordered Cartesian product operator but without the set membership
operator is decidable in unsorted sets.
We prove (see Corollary~\ref{cor:product}) that 
MLS with membership and Cartesian product operators is decidable in a sorted setting,
where all elements of a set have the same sort.

%
%
\section{Formal Preliminaries}\label{sec:preliminaries}
We define our theories and our calculi in the context of many-sorted logic 
with equality and polymorphic sorts and functions.
We assume the reader is familiar with the following notions from that logic:
signature, term, 
formula, free variable,
interpretation, and satisfiability in an interpretation.
Let \(\Sigma\) be a many-sorted signature.
We will denote sort parameters in polymorphic sorts with \(\alpha\) and \(\beta\),
and monomorphic sorts with \(\tau\).
We will use \(\teq\) as the (infix) logical symbol for equality --- which
has polymorphic rank \(\alpha \times \alpha\) and is always interpreted 
as the identity relation over \(\alpha\).
We assume all signatures \(\Sigma\) contain the Boolean sort \(\sbool\),
always interpreted as the binary set \(\{\mathit{true},\mathit{false}\}\),
and two Boolean constant symbols, \(\ltrue\) and \(\bot\), 
for \(\mathit{true}\) and \(\mathit{false}\).
Without loss of generality, we assume \(\teq\) is the only predicate symbol
in \(\Sigma\), as all other predicates can be modeled as functions
with return sort \(\sbool\).
We will write, e.g., \(p( x )\) as shorthand for \(p( x ) \teq \ltrue\),
where \(p( x )\) has sort \(\sbool\).
We write \(s \tneq t \) as an abbreviation for \(\lnot\,s \teq t\).
\begin{report}
When talking about sorts, we will normally say just ``sort'' to mean monomorphic sort,
and say ``polymorphic sort'' otherwise.
\end{report}
A \define{\(\Sigma\)-term/formula} is a well-sorted term/formula 
all of whose function symbols are from \(\Sigma\).
%
If \(\varphi\) is a \(\Sigma\)-formula and \(\I\) 
is a \(\Sigma\)-interpretation,
we write \(\I \models \varphi\) if \(\I\) satisfies \(\varphi\).
If \(t\) is a term, we denote by \(\I(t)\) the value of \(t\) in \(\I\).
A \define{theory} is a pair \(\T = (\Sigma, \Mo)\), where
\(\Sigma\) is a signature and  \(\Mo\) is a class of \(\Sigma\)-interpretations,
the \define{models} of \(\T\),
that is closed under variable reassignment
(i.e., every \(\Sigma\)-interpretation that differs from one in \(\Mo\)
only in how it interprets the variables is also in \(\Mo\)).

A \(\Sigma\)-formula \(\varphi\) is
\define{satisfiable} (resp., \define{unsatisfiable}) in \(\T\)
if it is satisfied by some (resp., no) interpretation in \(\Mo\).
%
%
Two $\Sigma$-formulas $\varphi_1$ and $\varphi_2$ are \define{equisatisfiable in \(\T\)}
if for every model \(\str A\) of \(\T\) that satisfies $\varphi_1$, there is
a model of \(\T\) that satisfies $\varphi_2$ and differs from \(\str A\) at most
in how it interprets the free variables not shared by $\varphi_1$ and $\varphi_2$.
%

\subsection{A Theory of Finite Relations}
We define a many-sorted theory \(\relationsTheory\) of finite relations.
Its signature \(\relationsSig\) is given in~\cref{fig:relations_sig}.
We use \(\alpha\) and \(\beta\), possibly with subscripts, as sort parameters
in polymorphic sorts.
Additionally, \(\relationsTheory\) has three classes of sorts,
with a corresponding polymorphic sort constructor:
predicate sorts, tuple sorts, and set sorts.
\define{Predicate sorts} are monomorphic instances of
\(\alpha_1 \times \dots \times \alpha_k \rightarrow \bool\) for all \(k \geq 0\).
\define{Tuple sorts} are constructed by the variadic constructor \(\tuple\)
which takes zero or more sort arguments.
For each \(k \geq 0\),
\(\tuple (\tau_1, \ldots, \tau_k)\) denotes the set of tuples of size \(k\)
with elements of sort \(\tau_1, \ldots, \tau_k\), respectively.
\define{Set sorts} are monomorphic instances \(\s(\tau)\) of \(\s(\alpha)\).
The sort \(\s(\tau)\) denotes the set of all \emph{finite} sets
of elements of the domain denoted by sort \(\tau\).
We model relations as sets of tuples and
write \(\relation(\tau_0, \ldots, \tau_k)\) as a shorthand for the sort
\(\s(\tuple (\tau_0, \ldots, \tau_k))\).
%

Signature \(\relationsSig\), summarized in Table~\ref{fig:relations_sig},
contains a subset of the set symbols from Mohamed et al.~\cite{lparBags}. 
%
%
The semantics of the various set operators is the expected one.
The operator \(\sqin\) denotes the set membership relation 
and \(\sqsubseteq\) the set inclusion relation. 
Expressions
\(\opsingleton{e}, s \sqcup t, s \sqcap t, s \setminus t, s \product t\)
denote the singleton set containing \(e\) and 
the union, intersection, difference,  and flat Cartesian product 
of sets \(s\) and \(t\), respectively.
The filter operator \(\sigma\) is a second-order operator 
taking a predicate as its first argument.
Terms \(\filter(p, s)\) denote the set consisting of the elements
of set \(s\) that satisfy predicate \(p\).
We extend the language to allow lambda abstractions and their applications, and
we allow the argument \(p\) of \(\filter(p, s)\) to be 
either a (second-order) variable
or a lambda abstraction,
both of rank \(\alpha \to \bool\)
if \(s\) has sort \(\s(\alpha)\).
For increased readability, we write \(\filter(p, s)\)
in the set comprehension notation as \(\compr{x}{x \sqin s \wedge p(x)}\).

The signature contains also
the set quantifiers \(\setAll\) and \(\setSome\), which we discuss in detail 
in~\cref{sec:all_some}.

In practice, we are not interested in theory \(\relationsTheory\) in isolation 
but in combination with some theory \(\elemsTheory\) of set/relation elements
(e.g., a theory of integers, strings, and so on).
For the rest of the paper then, we fix an extension \(\relationsTheory'\) 
of \(\relationsTheory\), with signature \(\relationsSig'\), 
that incorporates a (possibly combined) theory of elements
whose signature \(\elemsSig\) shares no function symbols with \(\relationsSig\).
Moreover, we consider only formulas none of whose terms have a sort
where \(\s\) is nested in some other constructor or itself.
This is a \emph{genuine restriction} that disallows, for instances,
sorts like \(\s(\s(\int))\) or \(\mathsf{List}(\s(\int))\).
Extending our results so as to drop this restriction is left to future work.

\begin{table}[!tbp]
  \caption{
    Signature \(\relationsSig\) for the theory of relations.}
    \label{fig:relations_sig}
  \begin{threeparttable}[b]
  \scriptsize
  \[
    \begin{array}{l}
      \begin{array}{@{}l@{\quad}l@{\quad}l@{\quad}l@{}}
        \toprule
        \textbf{Symbol}                                                                      & \textbf{Signature}                                                                                     & \textbf{SMT-LIB Syntax}         \\
        \midrule
        \sempty                                                                              & \s(\alpha)                                                                                         & \verb|set.empty|                 \\
        \opsingleton{\text{-}}                                                               & \alpha \rightarrow \s(\alpha)                                                                      & \verb|set.singleton|             \\
        \sqcup                                                                               & \s(\alpha) \times \s(\alpha) \rightarrow \s(\alpha)                                                & \verb|set.union|                 \\
        \sqcap                                                                               & \s(\alpha) \times \s(\alpha) \rightarrow \s(\alpha)                                                & \verb|set.inter|                 \\
        \setminus                                                                            & \s(\alpha) \times \s(\alpha) \rightarrow \s(\alpha)                                                & \verb|set.minus|                 \\
        \sqin                                                                                & \alpha \times \s(\alpha) \rightarrow \bool                                                         & \verb|set.member|                 \\
        \sqsubseteq                                                                          & \s(\alpha) \times \s(\alpha) \rightarrow \bool                                                     & \verb|set.subset|                  \\
        %
        \sigma                                                                               &
        \left(\alpha \rightarrow \bool\right) \times \s(\alpha) \rightarrow \s(\alpha)       & \verb|set.filter|                                                                                                                       \\
        \setAll                                                                              & \left(\alpha \rightarrow \bool\right) \times \s(\alpha) \rightarrow \bool                          & \verb|set.all|                      \\
        \setSome                                                                             & \left(\alpha \rightarrow \bool\right) \times \s(\alpha) \rightarrow \bool                          & \verb|set.some|                     \\
        \midrule
        \tup{ \ldots}                                                                        &
        \alpha_0 \times \dots \times \alpha_k \rightarrow \tuple(\alpha_0, \ldots, \alpha_k) & \verb|tuple|                                                                                                                       \\

        %
        \product                                                                             & \relation(\bm{\alpha}) \times \relation(\bm{\beta}) \rightarrow \relation(\bm{\alpha}, \bm{\beta})\tnote{1} & \verb|rel.product|       \\
        \bottomrule
      \end{array}
      \\[4ex]
    \end{array}
  \]

  \begin{tablenotes}
    \item [1]   Here \(\relation(\bm{\alpha}, \bm{\beta})\) is a shorthand for \(\relation(\alpha_0, \ldots, \alpha_p, \beta_0, \ldots, \beta_q)\)
    when \(\bm{\alpha} = \alpha_0, \ldots, \alpha_p\) and \(\bm{\beta} = \beta_0, \ldots, \beta_q\).
  \end{tablenotes}
\end{threeparttable}
\end{table}

%

\subsection{A Calculus for \(\relationsTheory'\)}\label{sec:relations_calculus}
To reason about quantifier-free formulas in \(\relationsTheory'\) 
we adopt a variant of the calculus described by Mohamed et al.~\cite{lparBags}.
%
Without loss of generality
we assume that every set term in such formulas is \define{flat}, 
i.e., 
\(s, t\) are set variables in all terms of the form
\(e \sqin s, s \sqcup t, s \sqcap t, s \setminus t, s \product t, \filter(p, s)\).
We also assume that all equalities have the form \(x \teq t\) where \(x\) is
a (element or set) variable. 
%
%
%
To simplify the presentation, we further restrict our attention only
to sets and relations over the \emph{same} element domain, denoted generically 
in the following by the \define{element sort} \(\eleSort\).
Our results, however, do not require this restriction.

Thanks to the following lemma, we will further focus on just sets
of \(\relationsSig\)-literals, or \define{relation constraints}, 
and \(\elemsSig\)-literals or \define{element constraints}.

\begin{restatable}{lemma}{setsConstraints}
  \label{lem:setsConstraints}
  For every quantifier-free \(\relationsSig'\)-formula \(\varphi\),
  there are sets \(S_1,\ldots,S_n\) of relation constraints
  and sets \(E_1,\ldots,E_n\) of element constraints
  such that
  \(\varphi\) is satisfiable in \(\relationsTheory'\)
  iff \(S_i \cup E_i\) is satisfiable in \(\relationsTheory'\) for some \(i\in[1,n]\).
\end{restatable}

As a final simplification, we assume without loss of generality that
for every term \(t\) of sort 
\(\tuple(\eleSort, \ldots, \eleSort)\) 
occurring in one of the sets \(S_i\) above,
\(S_i\) also contains the constraint \(t \teq \tup{x_0, \ldots, x_k}\)
where \(x_0, \ldots, x_k\) are variables of sort 
\(\eleSort\).
\medskip

\subsubsection{Configurations and Derivation Trees}
The calculus operates on \define{configurations}. 
These are either the distinguished configuration \unsat or
pairs \(\conf{\Sc, \Ec}\) 
where \(\Sc\) is a set of 
constraints over relations
and 
\(\Ec\) is a set 
of constraints over their elements. 
%
%

A derivation rule takes a configuration and, if applicable to it,
generates one or more alternative configurations.
A derivation rule \define{applies} to a configuration \(C\)
if all the conditions in the rule's premises hold for \(C\) \emph{and}
the rule application is not redundant.
An application of a rule is \define{redundant} if it has a conclusion
where each component in the derived configuration is a subset of
the corresponding component in the premise configuration.
We assume that for rules that introduce fresh variables,
the introduced variables are identical whenever the premises triggering 
the rule are the same.
In other words, we cannot generate an infinite sequence of rule applications
by continuously using the same premises to introduce fresh variables.
A configuration other than \unsat is \define{saturated} 
if every possible application of a derivation rule 
to it is redundant.
%
A configuration \(\conf{\Sc, \Ec}\) is \define{satisfiable} in \(\relationsTheory'\)
if the set \(\Sc \cup \Ec\) is satisfiable in \(\relationsTheory'\).

A \define{derivation tree} is a finite tree
where each node is a configuration whose children, if any,
are obtained by a non-redundant application of a derivation rule
to the node.
A derivation tree \define{derives} from a derivation tree \(T\)
if it is obtainable from \(T\) by applying a derivation rule to one of \(T\)'s leaves.
It is \define{closed} if all of its leaves are \unsat.

As we show later,
a closed derivation tree with root \(\conf{\Sc, \Ec}\) is a proof that
\(\Sc \cup \Ec\) is unsatisfiable in  \(\relationsTheory'\).
In contrast, a derivation tree with root \(\conf{\Sc, \Ec}\) and
a saturated leaf 
is a witness that \(\Sc \cup \Ec\) is satisfiable in \(\relationsTheory'\).


\begin{figure*}[!tbp]
  \centering
  %
  \begin{prooftree}
    \RightLabel{\ruleEqUnsat}
    \AxiomC{\(t \not\opequal t \in \Scclosed\)}
    \UnaryInfC{\unsat}
    \DisplayProof \hskip 1em
    \RightLabel{\ruleSetUnsat}
    \AxiomC{\(x \sqin s \in \Sc^*  \quad x \not\sqin s \in \Sc^* \)}
    \UnaryInfC{\unsat}   
    \DisplayProof \hskip 1em
      \RightLabel{\ruleSetUnsat}
      \AxiomC{\(x \opin \opemptyset \in \Scclosed\)}
      \UnaryInfC{\unsat}   
    \end{prooftree}
  \begin{prooftree}
    \RightLabel{\ruleEIdent}
    \AxiomC{\(e_1,e_2 \in \ter{\Sc^*}\)}
    \AxiomC{\(e_1,e_2\) are terms of the same element sort}
    \BinaryInfC{
      \(\Sc := \Sc, e_1 \teq e_2 ~~~
      \Ec := \Ec, e_1 \teq e_2
      ~~~\vert\vert~~~
      \Sc := \Sc, e_1 \tneq e_2 ~~~
      \Ec := \Ec, e_1 \tneq e_2
      \)}
  \end{prooftree}
  \smallskip
  
  \begin{tabular}{c}
    %
    \inferR[\ruleInterUp]
    {x \opin s \in \Scclosed \quad
      x \opin t \in \Scclosed  \quad
      s \opinter t \in \termsof{\Sc} }
    {\Sc := \cadd{\Sc}{x \opin s \opinter t}}
    \qquad
    \inferR[\ruleInterDown]
    {x \opin s \opinter t \in \Scclosed}
    {\Sc := \cadd{\cadd{\Sc}{x \opin s}}{x \opin t}}{}
    \\[2ex]
    \inferR[\ruleUnionUp]
    {x \opin u \in \Scclosed \quad
      u \in \{s,t\}            \quad
      s \opunion t  \in \termsof{\Sc}}
    {\Sc := \cadd{\Sc}{x \opin s \opunion t}}
    \qquad
    \inferR[\ruleUnionDown]
    {x \opin s \opunion t \in \Scclosed}
    {\Sc := \cadd{\Sc}{x \opin s}
      \ \parallel\
      \Sc := \cadd{\Sc}{x \opin t}}
    \\[2ex]
    
   \inferR[\ruleDifferenceUp]
   {x \opin s \in \Scclosed \quad
     s \opsetminus t \in \termsof{\Sc} }
   {\Sc := \cadd{\Sc}{x \opin t}
     \ \parallel\
     \Sc := \cadd{\Sc}{x \opin s \opsetminus t}}
   \quad
   \inferR[\ruleDifferenceDown]
   {x \opin s \opsetminus t \in \Scclosed}
   {\Sc := \cadd{\cadd{\Sc}{x \opin s}}{x \opnotin t}}
   \\[2ex]
    \inferR[\ruleSingleUp]
    {\opsingleton{x} \in \termsof{\Sc}}
    {\Sc := \cadd{\Sc}{x \opin \opsingleton{x}}}
    \qquad
    \inferR[\ruleSingleDown]
    {x \opin \opsingleton{y} \in \Scclosed}
    {\Sc := \cadd{\Sc}{x \opequal y}}
    \\[2ex]    
    \inferR[\ruleSetDiseq]
    {s \not\opequal t \in \Scclosed
    }
    {\Sc := \cadd{\cadd{\Sc}{z \opin s}}{z \opnotin t}
      \quad\parallel\quad
      \Sc := \cadd{\cadd{\Sc}{z \opnotin s}}{z \opin t}}
    \qquad
    \inferR[\ruleEConf]
    {\Ec \text{ is unsatisfiable in } \relationsTheory'}
    {\unsat}
    \\[2ex]
    \inferR [\ruleProductUp]
    {\tup{x_1, \ldots, x_m} \sqin r_1 \in \Scclosed \quad
      \tup{y_1, \ldots, y_n} \sqin r_2 \in \Scclosed \quad
      r_1 \opprod r_2 \in \termsof{\Sc}
    }
    {\Sc := \cadd{\Sc}{\tup{x_1, \ldots, x_m, y_1, \ldots, y_n} \sqin r_1 \opprod r_2}}
    \\[2ex]
    \inferR [\ruleProductDown]
    {\tup{x_1, \ldots, x_m, y_1, \ldots, y_n} \sqin r_1 \opprod r_2 \in \Scclosed \quad
      \arity(r_1) = m}
    {\Sc := \cadd{\cadd{\Sc}{\tup{x_1, \ldots, x_m} \sqin r_1}}{\tup{y_1, \ldots, y_n} \sqin r_2}}
  \end{tabular}
  %
  \begin{prooftree}
    \RightLabel{\rn{Filter up}}
    \AxiomC{\(e \sqin s \in \Sc^*\)}
    \AxiomC{\(\filter(p, s) \in \ter{\Sc}\)}
    \BinaryInfC{
      \(\Ec:= \Ec, p(e) \quad \Sc := \Sc, e \sqin \filter(p, s)
      ~~~\vert\vert~~~
      \Ec:= \Ec, \neg p(e)\quad \Sc := \Sc, e \not\sqin \filter(p, s)\)
    }
  \end{prooftree}
  \begin{prooftree}
    \AxiomC{\( e \sqin \filter(p, s) \in \Sc^*\)}
    \RightLabel{\rn{Filter down}}
    \UnaryInfC{\(\Ec := \Ec, p(e) \quad \Sc := \Sc, e \sqin s\)
    }
  \end{prooftree}

  \caption{The derivation rules. 
  Variable \(z\) in \ruleSetDiseq is fresh.
  In \ruleProductDown, \(\arity(r_1)\) denotes the arity of relation \(r_1\).}
  \label{fig:set_rules}
\end{figure*}
\medskip

\subsubsection{The Derivation Rules}
The rules of our calculus are listed in Figure~\ref{fig:set_rules}.
They are expressed in \define{guarded assignment form} where
the premise describes the conditions on the current configuration
under which the rule can be applied, and
the conclusion is either \unsat, or otherwise describes
\emph{only the changes} to the current configuration.
%
%
Rules with two conclusions, separated by the symbol \(\vert \vert\),
are non-deterministic branching rules, in the style of analytic tableaux.

In the rules, we write \(\Sc, c\), as an abbreviation of \(\Sc \cup \{c\}\) and denote by
\(\mathcal{T}(\Sc)\) the set of all terms and subterms occurring in \(\Sc\).
Because of our focus on the theory \(\relationsTheory\) in this paper,
the calculus relies on an \define{element oracle}
that can decide the satisfiability in \(\relationsTheory'\)
of sets of element constraints.
We also require the computability of the all predicates
used as arguments in applications of filter (\(\sigma\)).

We define the following closures for \(\Sc\) and \(\Ec\) where
\(\models_{\text{tup}}\) denotes entailment in the
theory of tuples, which treats all function symbols other than \(\tup{\_}\) 
as uninterpreted.

{\small
\begin{align*}
  \mathcal{\Sc}^* = \
   & \{(\lnot)s \teq t \mid s, t \in \ter{\Sc};\, \Sc \models_{\text{tup}} (\lnot)s\teq t\}                   \\
  \cup \
   & \{e \sqin s \mid e, s \in \ter{\Sc};\,
  \Sc \models_{\text{tup}} e \teq e' \wedge s \teq s', e' \sqin s' \in \Sc\} 
  \\[1ex]
  \mathcal{\Ec}^* = \ 
   & \{(\lnot)s \teq t \mid s, t \in \ter{\Ec};\, \Ec \models_{\text{tup}} (\lnot)s\teq t\}                   \\
  \cup \
   & \{(\lnot)p(e) \mid (\lnot) p(e') \in \Ec; e, e' \in \ter{\Ec};\, \Ec \models_{\text{tup}} e \teq e'\} \\
\end{align*}
}

\noindent
In the definitions above, the primed variables are implicitly existentially quantified.
Note that \(\Sc^* \supseteq \Sc\) and \(\Ec^* \supseteq \Ec\).

The sets \(\Sc^*\) and \(\Ec^*\) are computable 
from \(\Sc\) and \(\Ec\) by an extension of standard congruence closure procedures
with 
%
rules for deducing the equalities \(s_1 \teq t_1, \ldots, s_n \teq t_n\) 
from tuple equalities of the form
\(\tup{s_1, \dots, s_n} \teq  \tup{t_1, \dots, t_n}\).

Moving to the derivation rules,
rule \ruleEConf is applied when a conflict is found by the element solver.
The other rules generating \(\bot\) should be self-explanatory.
%
%
%
%
Rule \ruleSetDiseq handles disequality between two sets \(s,t\)
by stating that some element, represented by a fresh variable \(z\),
only occurs in \(s\) or in \(t\), but not both.
Rules \ruleUnionUp, \ruleUnionDown, \ruleInterUp, and
\ruleInterDown 
correspond directly to the semantics of their operators.
We omit for brevity an upward and a downward rule for set difference 
since they are similar.
\ruleTableProductUp and \ruleTableProductDown are upward and downward rules
for the \(\product\) operator.
\ruleSetFilterUp splits on whether an element \(e\) in \(s\) satisfies
(the predicate denoted by) \(p\) or not in order to determine its membership
in set \(\filter(p, s)\).
\ruleSetFilterDown concludes that every element in \(\filter(p, s)\)
necessarily satisfies \(p\) and occurs in \(s\).
%

%

\section{Decidability of Restricted Filter}\label{sec:decidable}
%
%
%
%

We prove that, under the conditions below,
the calculus we have introduced is sound, complete and terminating.

\begin{condition} \label{condI}
No predicate in applications of the filter operator \(\filter\) 
includes \define{set terms}, terms with sorts of the form \(\s(\tau)\).
\end{condition}

\begin{condition} \label{condII}
The satisfiability of sets of element constraints is decidable.
\end{condition}

\begin{condition} \label{condIII}
The satisfiability of sets \(\Ec\) of elements constraints 
remains decidable with the addition of predicate applications generated 
by rules \ruleSetFilterUp and \ruleSetFilterDown.
\end{condition}
Note that Condition~\ref{condIII} follows from~\ref{condII}
when all filter predicates are expressed as $\lambda$-terms in the language of \(\Ec\).


\begin{restatable}[Derivation]{definition}{def:derivation}
  \label{def:derivation}
  Let \(S = \Sc_0 \cup \Ec_0\) be a set of \(\relationsTheory\)-constraints.
  A \define{derivation} of \(S\) is a sequence \((T_i)_{0 \leq i < \kappa}\)
  of derivation trees, with \(\kappa\) finite or countably infinite,
  such that \(T_{i+1}\) derives from \(T_i\) for all \(i\)
  and \(T_0\) is a singleton tree with root \(\conf{\Sc_0, \Ec_0}\).
  A \emph{refutation} of \(S\) is a finite derivation of \(S\) 
  that ends with a closed tree.
\end{restatable}
%

A derivation strategy is \define{progressive} if it halts 
only with a closed tree or one with a saturated configuration.
A calculus is \define{terminating} if
every progressive derivation strategy for it eventually halts.
%

Let \(\filterDecidable\) be the sublanguage of constraints in \(\relationsTheory\)
that satisfy Condition~\ref{condI} and contain set operators exclusively
from \(\{\teq, \sqin, \sqcup, \sqcap, \setminus, \product, \filter\}\).
Note that the omission of \(\sqsubseteq\) is no real restriction 
since \(s \sqsubseteq t\) can be expressed as \(s \teq s \sqcap t\).
We refer to constraints from \(\filterDecidable\) as
\define{\(\filterDecidable\)-constraints} and restrict our attention to them
in this section.
 
%

\subsection{Termination}
To start, we argue that the applicability of every rule in~\cref{fig:set_rules} is decidable.
This is trivially the case for rule \ruleEConf thanks to 
our Conditions~\ref{condII} and~\ref{condIII}.
For the other rules, the argument is fairly straightforward.
For example, \(\Sc^*\) can be computed from \(\Sc\) by a congruence closure algorithm
so checking the presence of constraints in \(\Sc^*\) is decidable.
%
%
Therefore, we focus on proving that the calculus has no infinite derivations
when starting with a finite set of \(\filterDecidable\)-constraints.


\renewcommand{\arraystretch}{1.3}
\begin{table*}
  \begin{center}
    \caption{Ranking functions for relation rules.
    Here, \(\vec{x}=\tup{x_1, \dots, x_m}\),
    \(\vec{y}=\tup{y_1, \dots, y_n}\),
    and 
    \(\tup{\vec{x}, \vec{y}}\)
     denotes \(\tup{x_1, \dots, x_m, y_1, \dots, y_n}\).
    }
  \label{fig:set_ranking}
    \begin{tabular}{l|l|l}
      \toprule
      \(f_i\)    & Rule                & Definition                                                                                                                       \\
      \midrule
      \(f_1, f_2\)    & \ruleInterUp, \ \ruleInterDown        & \(e_2 \cdot s_0^2 - \card{\{x \sqin s \sqcap t \mid x \sqin s \sqcap t \in S\}} \)
      , \quad
      %
      \(e_2 \cdot s_0^2  - \card{\{x \sqin s \mid x \sqin s \in S\}} \)                                                             \\
      \(f_3, f_4\)    & \ruleUnionUp, \ \ruleUnionDown       & \(e_2 \cdot s_0^2 - \card{\{x \sqin s \sqcup t \mid x \sqin s \sqcup t \in S\}} \)
      , \quad
      %
      \(e_2 \cdot s_0^2  - \card{\{x \sqin s \mid x \sqin s \in S\}} \)                                                             \\
      \(f_5, f_6\)    & \ruleDifferenceUp, \ \ruleDifferenceDown   & \(e_2 \cdot s_0^2  - \card{\{x \sqin s \mid x \sqin s \in S\}} \)
      , \quad
      %
      \(e_2 \cdot s_0^2  - \card{\{x \not\sqin s \mid x \not\sqin s \in S\}} \)                                                     \\
      \(f_7, f_8\)    & \ruleSingleUp, \ \ruleSingleDown       & \(s_0  - \card{\{x \sqin [x] \mid x \sqin [x] \in S\}} \)
      , \quad
      %
      \(e_2 \cdot s_0  - \card{\{x \teq y \mid x \teq y, x \sqin [y] \in S\}} \)                                                    \\
      \(f_9\)    & \ruleSetDiseq       & \(s_0^2  - \card{\{\tup{ z_{s,t} \sqin s, z_{s,t} \not\sqin t} \mid z_{s,t} \sqin s,z_{s,t} \not\sqin t, s \tneq t \in S\}} \) \\
      \(f_{10}, f_{11}\) & \ruleProductUp, \ruleProductDown      & \(e_2^2 \cdot s_0^2  - \card{\{\tup{ \vec{x}, \vec{y}} \mid \tup{\vec{x}, \vec{y}} \sqin s\product t \in S\}} \)
      , \quad
      %
      \(e_2^2 \cdot s_0^2  - \card{\{\tup{ \vec{x}, \vec{y}} \mid \vec{x} \sqin s \in S, \vec{y} \sqin t \in S\}} \)              \\
      %
      %
      \(f_{12}\) & \ruleEIdent          & \(e_2^2 - \card{\{e_1 \teq e_2 \mid e_1 \teq e_2 \in S\}} - \card{\{e_1 \tneq e_2 \mid e_1 \teq e_2 \in S\}} \)                \\
      \(f_{13}\) & \ruleSetFilterUp    & \(
      \hspace{-.65em}
      \begin{array}{l@{\,}c@{\,}l}
        e_2 \cdot s_0 & - & \card{\{\tup{p(e), e \sqin \filter(p, s)}  \mid p(e) \in E, e \sqin \filter(p, s) \in S \}}                   \\
                       & - & \card{\{\tup{\neg p(e), e \not\sqin \filter(p, s)}  \mid \neg p(e) \in E, e \not\sqin \filter(p, s) \in S \}}
      \end{array}
      \)                                                                                                                                                                  \\
      \(f_{14}\) & \ruleSetFilterDown  & \(
      e_2 \cdot s_0 - \card{\{\tup{p(e), e \sqin s} \mid p(e) \in E, e \sqin s \in S\}}
      \)                                                                                                                                                                  \\
      \bottomrule
    \end{tabular}
  \end{center}
\end{table*}

\begin{restatable}{lemma}{lem:relationsProductTerminating}
  \label{lem:relations_product_terminating}
  Under Conditions \ref{condI}--\ref{condIII},
  every derivation of a set of \(\filterDecidable\)-constraints 
  is finite.
\end{restatable}
\begin{proof}[Proof sketch]
  %
  %
  Suppose we start with a configuration \(C = \conf{\Sc,\Ec}\) in \(\filterDecidable\).
  For uniformity, and without loss of generality, suppose all set terms denote
  relations, i.e., sets whose elements are tuples.
  Now let \(e_0 = |\ter{\Ec}|\) and \(s_0 = |\ter{\Sc}|\)
  where the sets \(\ter{\Ec}\) and \(\ter{\Sc}\) consist of
  all the finitely-many terms of sort \(\tupleSort\) and \(\s\)
  in \(\Sc\) and \(\Ec\), respectively.
  Since none of the derivation rules introduces new \(\s\) terms,
  \(s_0\) is constant in all derivation trees with root \(C\).
  %
  %
  None of the rules except for \ruleSetDiseq, \ruleProductUp and \ruleProductDown
  generate new element terms.
  (Rules \ruleSetFilterUp and \ruleSetFilterDown do not introduce new element terms
  because of our assumption that the predicates do not include set terms.)
  Since \ruleSetDiseq can be applied only once for each disequality constraint,
  %
  %
  \(e_0\) can increase by at most \(s_0^2\)
  in all derivation trees. 
  Let \(e_1 = e_0 + s_0^2\).
  Rules \ruleProductUp and \ruleProductDown introduce up to 
  \(e_1^2\) and \(2 e_1\) new element terms, respectively.
  This means that the number of element terms in all derivation is at most 
  \(e_2 = e_1^2 + 2e_1 + e_1 = e_1^2 + 3e_1\).
  %
    
  We now show that every derivation starting from \(C\) is finite.
  We do that by defining a well-founded order \(\succ\)
  on the set of configurations and prove that applying each 
  rule to a configuration \(X = \conf{S,E}\) yields a configuration \(Y\)
  such that \(X \succ Y\). 
  The order is defined as follows: \(X \succ Y\) iff
  \begin{enumerate}
    \item \(X \neq \unsat\) and \(Y = \unsat\); or
    \item \(X \neq \unsat\), \(Y \neq \unsat\) and
          \(
          \tup{f_1(X), \dots, f_{\terminatingSetCount}(X)} \succ_\mathrm{p}
          \tup{f_1(Y), \dots, f_{\terminatingSetCount}(Y)}
          \)
          where \(f_i\) are ranking functions defined in~\cref{fig:set_ranking}
          and \(\succ_\mathrm{p}\) is the pointwise ordering extension of 
          \(>\) over the integers.
  \end{enumerate}
  Since the conflict rules \ruleEConf, \ruleSetUnsat, \ruleEmptyUnsat, and \ruleEqUnsat
  derive the configuration \(\unsat\), which is a minimal element of $\succ$,
  the claim \(X \succ Y\) is immediate for them.
  Each application of the other rules only reduces the value 
  of its corresponding ranking function in~\cref{fig:set_ranking}
  and leaves the value of the other functions unchanged.
  %
  %
  For these rules too we then have that \(X \succ Y\).
  To conclude the proof we only have to argue that \(\succ\) is well-founded
  but this is a consequence of the fact that, in any given derivation tree,
  every function in~\cref{fig:set_ranking} is bounded below by 0.
\end{proof}

Termination is a direct consequence of Lemma~\ref{lem:relations_product_terminating}.

\begin{restatable}{proposition}{lem:setFilterTerminating}
  \label{lem:setFilterTerminating}
   Under Conditions \ref{condI}--\ref{condIII},
   the calculus is terminating over sets of \(\filterDecidable\)-constraints.
\end{restatable}
%

\subsection{Refutation Soundness}

The refutation soundness of the calculus follows from the fact
that all of its derivation rules preserve equisatisfiability.

\begin{restatable}{lemma}{setFilterSound}
  \label{lem:setFilterSound}
  For every rule of the calculus, 
  the premise configuration is equisatisfiable in \(\relationsTheory\) 
  with the disjunction of its conclusions.
\end{restatable}
%
%
\begin{restatable}[Soundness]{proposition}{prop:setFilterSound}
  \label{prop:setFilterSound}
  Every set of \(\filterDecidable\)-constraints that has a refutation is
  unsatisfiable in \(\relationsTheory\).
\end{restatable}
\begin{proof}
  Given that, by Lemma~\ref{lem:setFilterSound}, every rule preserves
  constraint equisatisfiability,
  we can argue by structural induction on derivation trees that
  the root of any closed derivation tree is unsatisfiable in \(\relationsTheory\).
  The claim then holds because every refutation of a set
  \(S\) of \(\filterDecidable\)-constraints starts with a configuration 
  that is equisatisfiable with \(S\) in \(\relationsTheory\).
\end{proof}

%

\subsection{Refutation Completeness}

Like most calculi, ours is not refutation complete 
for \emph{arbitrary} derivation strategies, 
even if restricted to \(\filterDecidable\)-constraints.
However, it is complete with progressive ones.

\begin{restatable}[Completeness]{proposition}{lem:setFilterComplete}
  \label{lem:setFilterComplete}
  Under Conditions \ref{condI}--\ref{condIII},
  if \(S = \Sc_0 \cup \Ec_0\) is a set of \(\filterDecidable\)-constraints
  unsatisfiable in \(\relationsTheory\),
  every derivation of \(S\) generated with a progressive derivation strategy 
  extends to a refutation.
\end{restatable}
\begin{proof}
  We prove the contrapositive.
  Suppose \(S\) has a derivation \(D\) that cannot be extended to a refutation.
  Since the calculus is terminating (Lemma~\ref{lem:setFilterTerminating}),
  \(D\) must be extensible, by progressiveness, to a derivation that ends
  with a tree with a saturated leaf \(\conf{\Sc_\mathrm{l}, \Ec_\mathrm{l}}\).
  By Lemma~\ref{prop:setModel} (below),
  \(\Sc_\mathrm{l} \cup \Ec_\mathrm{l}\) is then satisfiable in \(\relationsTheory\).
  The satisfiability of \(S\) in \(\relationsTheory\) follows from that fact
  that \(S \subseteq \Sc_\mathrm{l} \cup \Ec_\mathrm{l}\), which can be proven 
  by structural induction on derivation trees.
\end{proof}

The main step in the proof of the proposition above relies on the following result
where
\(\vars(A)\) denotes all the variables occurring in some expression in the set \(A\).

\begin{restatable}{lemma}{prop:setModel}
  \label{prop:setModel}
  Let \(\Sc_0 \cup \Ec_0\) be a set of \(\filterDecidable\)-constraints 
  where \(\Sc_0\) and \(\Ec_0\) contain set and element constraints respectively.
  Suppose \(D\) is a finite derivation of \(\conf{\Sc_0, \Ec_0}\).
  If its final tree has a saturated leaf 
  \(C_\mathrm{l} = \conf{\Sc_\mathrm{l}, \Ec_\mathrm{l}}\),
  then there exists a model \(\I\) of \(\relationsTheory\) 
  that satisfies \(\Sc_\mathrm{l} \cup \Ec_\mathrm{l}\) and
  has the following properties:
  \begin{enumerate}
    \item \label{prop:one}
     For all \(x, y \in \vars(\Sc_\mathrm{l}) \cup \vars(\Ec_\mathrm{l})\) 
     of element sort, \\
     \(\I(x) = \I(y)\) iff \(x \teq y \in \Ec_\mathrm{l}^*\).
    \item \label{prop:two} 
     For all \(s \in \vars(\Sc_\mathrm{l})\) of set sort, \\
     \(\I(s) = \{\I(x) \mid x \sqin s \in \Sc_\mathrm{l}^*\}\).
  \end{enumerate}
\end{restatable}
\begin{proof}[Proof sketch]

  For simplicity, 
  we consider only initial configurations
  \(\conf{\Sc_0, \Ec_0}\) that
  contain no variables of sorts of the form \(\tupleSort(\alpha_1, \dots, \alpha_n)\)
  since any such variable can be replaced by a tuple \(\tup{x_1, \dots, x_n}\) 
  where each \(x_i\) is a variable of sort \(\alpha_i\).
  We assume, with no loss of generality, that all set equalities in
  \(\Sc_\mathrm{l}\) have the form \(x \teq t\) where \(x\) is a variable 
  that never occurs in the right-hand side of an equality of \(\Sc_\mathrm{l}\).
  This restriction allows us to define a well-founded ordering 
  \(\succ_\mathrm{T}\) over the set variables in \(\Sc_\mathrm{l}\)
  where \(x \succ_\mathrm{T} y\) iff \(x \teq t \in \Sc_\mathrm{l}\) and
  \(y\) occurs in \(t\). 
  

  Since the leaf \(\conf{\Sc_\mathrm{l}, \Ec_\mathrm{l}}\) is saturated, 
  rules \ruleEIdent, \ruleSetFilterDown and \ruleSetFilterUp do not apply.
  This implies all equalities between elements and
  predicates generated by the filter rules
  have been propagated to component \(\Ec_\mathrm{l}\).
  Since rule \ruleEConf does not apply either, we can conclude that
  \(\Ec_\mathrm{l}\) is satisfiable in \(\relationsTheory'\).
  Let \(\I\) be any interpretation that satisfies \(\Ec_\mathrm{l}\).
  It is possible to show that \(\I\) satisfies Property~(\ref{prop:one}). 
  
  We modify \(\I\) so that it assigns to each set variable \(s\)
  in \(\Sc_\mathrm{l}\) the value \(\{\I(e) \mid e \sqin s \in \Sc_\mathrm{l}^*\}\).
  Since \(\Ec_\mathrm{l}\) contains no set variables, \(\I\) continues to satisfy
  \(\Ec_\mathrm{l}\).
  We argue by induction on \(\succ_\mathrm{T}\) 
  that \(\I\) satisfies every constraint \(c\) in \(\Sc_\mathrm{l}\) as well.
  %
  %
  %
  
  If \(c\) is a membership constraint,
  we know \(c\) has the form \(e \sqin s\) where \(e\) and \(s\) are variables.
  Then we have that \(\I(e) \in \I(s)\) by construction of \(\I\)
  because \(\Sc_\mathrm{l} \subseteq \Sc_\mathrm{l}^*\).
  If \(c\) is an equality, we know \(c\) has 
  either the form \(e_1 \teq e_2\) where \(e_1\) is an element variable 
  or the form \(s \teq t\) where \(s\) is a set variable.
  In the first case, thanks to rule \ruleEIdent and saturation
  \(e_1 \teq e_2\) is in \(\Ec_\mathrm{l}\) and so \(e_1 \teq e_2\) is satisfied
  by \(\I\) by Property~(\ref{prop:one}).
  In the second case, we have the following cases for \(t\):
  %
    
    1)
    \(s \teq t\) where \(t\) is a variable.
          We have \(\I(s) = \I(t)\) because then, for each element variable \(e\),
          \(e \sqin s \in \Sc_\mathrm{l}^*\) iff \(e \sqin t \in \Sc_\mathrm{l}^*\) 
          by the definition of \(\Sc_\mathrm{l}^*\).
          %
    
    2)
    \(s \teq \filter(p, t)\) where \(s\) and \(t\) are set variables.
          Observe that \(s \succ_\mathrm{T} t\) and let
          \(A = \{a \mid a \in \I(t), \I(p)(a) = \mathit{true}\}\).
          By the semantics of \(\sigma\),
          we need to show that \(\I(s) = A\).          
          Suppose \( a \in \I(s)\). 
          By the definition of \(\I\), we have 
          \(a = \I(e)\) for some  \(e\) with \(e \sqin s \in \Sc_\mathrm{l}^*\). 
          By the definition of \(\Sc_\mathrm{l}^*\), we have \(e \sqin \filter(p, t) \in \Sc_\mathrm{l}^*\). 
          Now \(e \sqin t \in \Sc_\mathrm{l}, p(e) \in \Ec_\mathrm{l}\) because of saturation wrt rule \ruleSetFilterDown. 
          By induction on \(t\) and the fact that \(\I\) satisfies \(\Ec_\mathrm{l}\), we have
           \( a  \in \I(t)\) and \(\I(p)(a)\) is true. 
          Therefore \(x \in A\). 
          For the other direction, suppose \(a \in A\). 
          That means \( a \in \I(t)\) and  \(\I(p)(a)\) is true. 
          From the induction hypothesis on \(t\), we have
          \( a = \I(e) \) for some \(e \sqin t \in \Sc_\mathrm{l}^*\). 
          By saturation wrt rule \ruleSetFilterUp, \(\Ec_\mathrm{l}\)
          must contain one of the rule's conclusions.
          The second one cannot be, 
          since \(\neg p(e) \in \Ec_\mathrm{l}\) implies \(\neg \I(p(e))\) 
          which means \(\I(p)(a)\) is false, a contradiction.
          So \(\Ec_\mathrm{l}\) must contain the first conclusion 
          of \ruleSetFilterUp.
          But then we have 
          \(p(e) \in \Ec_\mathrm{l}, e \sqin \filter(p, t) \in \Sc_\mathrm{l} \). 
          This means that \(a = \I(e), e \sqin s \in \Sc_\mathrm{l}^*\). 
          From the definition of \(\I(s)\), we have \(a \in \I(s)\).
          %

3)
  The remaining cases \(s \teq t \sqcup u, s \teq t \sqcap u, s \teq t \setminus u\) and \(s \teq t \product u\) 
  are proved by induction in a similar fashion. 

  We now argue that \(\I\) satisfies
  negative membership and disequality constraints.
  Let \(e \not \sqin s \in \Sc_\mathrm{l}\) and, by contradiction, suppose
  that \( \I(e) \in \I(s)\).
  Then, by construction of \(\I\), it must be the case that \(e \sqin s \in \Sc_\mathrm{l}^* \).
  But then rule \ruleSetUnsat appliest
  which contradicts our assumption that 
  \(\conf{\Sc_\mathrm{l},_\mathrm{l} \Ec_\mathrm{l}}\) is saturated.

  For constraints of the form \(s \tneq t \in \Sc_\mathrm{l}^*\),
  \(s\) and \(t\) could be either element terms or set terms.
  In the first case, suppose that \(\I(s) = \I(t)\),
  then it must be that \(s \teq t \in \Sc_\mathrm{l}^*\) 
  which makes \ruleEqUnsat applicable, again against the assumption that
  our leaf is saturated.
  If the second case, since rule \ruleSetDiseq does not apply 
  because the leaf is saturated, it must be that one of its conclusions,
  \(z \sqin s, z \not\sqin t\) or \(z \sqin t, z \not\sqin s\),
  is in \(\Sc_\mathrm{l}\).
  But then \(\I(z)\) cannot be in both \(\I(s)\) and \(\I(t)\), 
  which means that \(\I(s) \neq \I(t)\).
\end{proof}

%

\begin{restatable}{corollary}{cor:product}
  \label{cor:product}
  The satisfiability in \(\relationsTheory\) of \(\filterDecidable\)-constraints is decidable.
\end{restatable}

%

\subsection{Undecidability with Unrestricted Filter Predicates}
\label{sec:undecidable_set_fragments}

One may wonder if the restriction expressed by Condition~\ref{condI}
is really necessary for the decidability result in the previous section.
The answer is that some restriction on the predicates passed to filter
is indeed needed because otherwise decidability is lost.
We show that in this section at a high level by a 
sketching a reduction 
from a known undecidable problem.
To start, we show that having unrestricted filter predicates
allows us to define a set map operator \(\map\),
which takes a function \(f\) of rank \(\alpha \to \beta\) and
a finite set \(t\) of sort \(\s(\alpha)\), and
returns the image of \(t\) under \(f\).
We can do that in sets of constraints
by replacing every term of the form \(\map(f, t)\)
by a fresh variable \(s\) and adding the two constraints:
\[
  \begin{array}{l@{~}l@{~}l}
    t & \teq & \compr{x}{x \sqin t \land f(x) \sqin s} \\[0ex]  
    s & \teq & \compr{y}{y \sqin s \land \compr{x}{x \sqin t \wedge f(x) \teq y} \tneq \sempty}
  \end{array}
\]
i.e., in desugared notation:
{\small
\[
    t \teq \filter(\lambda x.\, f(x) \sqin s,  t) \qquad
    s \teq \filter(\lambda y.\, \filter(\lambda x.\, y \teq f(x), t) \tneq \sempty, s)
\]
}

%
%
The first constraint ensures that each element
in \(t\) has an image under \(f\) in \(s\),
whereas the second ensures that
each element in \(s\) has a preimage under \(f\)
in \(t\).
Next, we show the satisfiability of constraints with map terms
is undecidable.
%

%
The necessity of Condition~\ref{condII} and~\ref{condIII} for decidability
should be clear, so here we focus on the necessity of Condition~\ref{condI}.
We consider a particular element theory \(\elemsTheory\)
that satisfies conditions~\ref{condII} and~\ref{condIII},
and demonstrate the undecidability of \(\relationsTheory'\) 
when Condition~\ref{condI} is dropped.
Specifically, we consider as \(\elemsTheory\)
the combination of linear integer arithmetic
(\(\lian\)) and the theory of pairs, i.e., tuples of two elements.

The decidability of the satisfiability of quantifier-free formulas
in this combined theory can be proven by standard theory combination results~\cite{nelson}.
It is well-known that (\(\lian\)) has a decidable quantifier-free satisfiability problem and
is stably infinite\footnote{%
  A \(\Sigma\)-theory \(T\) is \define{stably infinite}
  over a sort \(\alpha\) in \(\Sigma\)
  if every quantifier-free \(\Sigma\)-formula satisfiable in \(T\) is satisfiable
  in a model of \(T\) that interprets \(\alpha\) as an infinite set.
}.
The same is true for the theory of pairs, which is stably infinite over each
of its element sorts.
This is a consequence of more general results about the theory
of algebraic datatypes~\cite{datatypes,selectorsPaper}.
As a result, the combination of these two theories,
the theory of integers and integer pairs with no multiplication,
has a decidable quantifier-free satisfiability problem.
We will use this combined theory
as the element theory for our theory of sets and show that
the satisfiability of set constraints containing applications of \(\map\)
to functions from \(\lian\) and pairs is undecidable.
We do that by reducing Hilbert's tenth problem to formulas
in a language we call \(\sigmap\) that
only accepts integers and integer pairs as elements, and only allows arithmetic operators inside
filter predicates.\footnote{%
Note for reviewers:
the full grammar for \(\sigmap\) is in the appendix. 
}
For Hilbert's tenth problem, it is enough to only consider the system of equations described in~\cite{undecidable_hilbert_set}:
\(  x \teq y ,  x \teq y + z,   x  \teq y \cdot z,  x  \teq k \),
where \(x, y , z\) are nonnegative integers and \(k\) is a constant.
Let \(\sigHilbert\) denote the language of these equations.
The syntax of \(\sigmap\) allows us to define \(x, y, z, k\) as elements (integers),
and use arithmetic expressions in the first argument of map terms, but it does not
allow us to use addition and multiplication between element variables.
Therefore, we need to encode them with other operators in \(\sigmap\).

We define for each variable \(x\) a set \(x'\) that is the singleton containing \(x\).
%
Then we add the two constraints: \(x' \teq \pi(\lambda n. \ite(n \geq 0, n, - n), x')\)
and \(x' \teq [x] \).
Any model \(\I\) that satisfies the first constraint ensures that \(x'\) only contains nonnegative integers, and
the second constraint ensures that \(x'\) is interpreted as the singleton \(\{\I(x)\}\).
For constraints in \(\sigHilbert\) of the form \(x \teq y\) and
\(x \teq k\), we add the constraints
\(x' \teq y'\) and \(x' = [k]\) in \(\sigmap\) respectively.
For each constraint in \(\sigHilbert\) of the form \(x \teq y + z\)
we add the constraint from \(\sigmap\):
\(x' \teq \pi(\lambda n. y + n, z')\).


For each constraint \(x \teq y \cdot z\) in \(\sigHilbert\) 
we add an integer variable \(v_{y,z}\), a set variable \(p_{yz}\),
and a clause
\((y \teq 0 \wedge x' \teq [0]) \vee 
  (z \teq 0 \wedge x' \teq [0]) \vee 
  (y \tneq 0 \wedge z \tneq 0 \wedge \varphi)\)
in \(\sigmap\)
where \(\varphi\) is the conjunction of the following constraints:
\vspace{-1.5ex}

{
\small
\begin{align*}
  x' & \teq [v_{y,z}] \qquad 
  [v_{y,z}] \teq \pi(\lambda (m,n). ite(m \teq 1, n, v_{y,z}) , p_{yz})
  \\
  p_{yz}  & \teq [(y,z)]  
  \sqcup 
  \pi(\lambda (a,b).\, \ite(a \teq 1, (a,b), (a-1, b + z)), p_{yz})
\end{align*}
}
Now, any model that satisfies the system of equations in \(\sigHilbert\) defines a model that satisfies the corresponding constraints
in \(\sigmap\) by interpreting variables as follows:
\vspace{-1.5ex}

{\small
\begin{align*}  
  \I(x')  = \{\I(x) \}
  \qquad\qquad
  \I(v_{y,z}) = \I(y) \cdot \I(z)
  \\
  \{(\I(y), \I(z)), (\I(y)-1, 2 \cdot \I(z)),\dots, (1, \I(y) \cdot \I(z))\}
  \subseteq \I(p_{yz})
\end{align*}
}

Note that \(\I(x)\) must be nonnegative due to our constraints.
Due to the undecidability of Hilbert's tenth problem~\cite{hilbert},
we can conclude that the satisfiability in \(\relationsTheory\)
of formulas in the fragment \(\sigmap\) is undecidable.

%
%
\section{Bounded Set Quantifiers}\label{sec:all_some}
Applications that can encode their problems as formulas in \(\relationsTheory\) 
typically require the use of quantifiers.
However, in most cases all quantified variables are \emph{bounded}
in the sense that are constrained to range over a specific finite set.

Now, the theory \(\relationsTheory\) has the nice property of being able 
to express bounded quantification at the quantifier-free level.
In fact, since formulas in this theory are just boolean terms,
every formula of the form 
\(\exists x.\, (x \sqin s \land \varphi)\)
can be expressed equivalently as the constraint
\(\filter(\lambda x.\, \varphi, s) \tneq \sempty\).
Similarly, every formula of the form 
\(\forall x.\, (x \sqin s \Rightarrow \varphi)\)
can be expressed equivalently as the constraint
\(\filter(\lambda x.\, \varphi, s) \teq s\).

Since the SMT solver cvc5 already supports an extension of \(\relationsTheory\),
we extended its language further with the set-bounded quantifier operators
\(\setSome\) and \(\setAll\)
defined internally as:
\begin{align*}
\setSome(p, s) \equiv  \filter(p, s) \tneq \sempty \qquad 
\setAll(p, s)  \equiv  \filter(p, s) \teq s
\end{align*}
for all predicates \(p\) of rank \(\alpha \to \bool\) and 
sets \(s\) of sort \(\s(\alpha)\).

Referring back to the decidable fragment \(\filterDecidable\)
defined in Section~\ref{sec:decidable}
notice that formulas of the form
\(\setAll(\dots(\setAll(p, s_n)) \dots, s_1)\)
fall outside that fragment.
However, since the theory \(\relationsTheory\) supports the cross-product operator
among sets (yielding a set of tuples) it is possible to rewrite
nested set-bounded quantifiers that are all universal into constraints of the form  
\(\setAll(p, s_1 \product \dots \product s_n)\)
where \(p\) ranges over tuples.
Such constraints do fall in \(\filterDecidable\),
which means that we can effectively express in \(\filterDecidable\), 
and hence decide, universal formulas with set-bounded quantifiers. 
This result is in line with the decidability of the (unsorted) logic \universals, 
which only allows set-bounded universal quantifiers~\cite{universals}.

Regardless of decidability considerations, the ability to express bounded
quantification in an SMT solver without actually using standard quantifiers
is rather enticing since the performance of SMT solvers is notoriously fickle 
in the presence of quantifiers in input formulas.

We investigated the potential of set-bounded quantifiers 
with an experimental evaluation 
that considers formulas generated by the \sleec tool.
We discuss this investigation next.

%
%
\section{Implementation and Experimental Evaluation}\label{sec:sleec}
\sleec is a formal language for writing Social, Legal, Ethical,
Empathetic, and Cultural requirements~\cite{sleec}.
It aims to formalize requirements in such domains as customer service, healthcare,
and education.
It provides an intuitive high-level language for writing requirements, and
uses automated tools to identify conflicts, redundancies, and concerns.
\sleec is an event-based rule language. 
For example, a simple \sleec rule ``\textbf{when} $A$ \textbf{then} $B$ \textbf{within} $30$ seconds''
specifies the constraint that whenever 
an event of type $A$ occurs, there must 
be some occurrence of event $B$ within 30 seconds
 from the time of $A$'s occurrence. 
 $A$ is called the \define{rule trigger} and $B$ the \define{response}.

The tool LEGOS-\sleec translates requirements into the logic FOL*
of quantified formulas over relational objects, where a relational object with
$n$ attributes of a relational class
$C$ represents an
$n$-ary tuple in a set $s_{C}$~\cite{sleecEncoding}. 
The tool checks the satisfiability of FOL* formulas using LEGOS, 
a custom-built bounded satisfiability checker~\cite{legos}.
LEGOS incrementally expands the quantified formula within a fixed domain 
of relational objects,
generating quantifier-free over- and under-approximations of the formula. 
The satisfiability of these approximations is determined using the SMT
solver z3.
When the over-approximation is satisfiable but the under-approximation
is unsatisfiable, LEGOS computes a minimal correction set for the
under-approximation based on the difference between the two
approximations. 
This correction set is used to incrementally
expand the domain of the bounded variables, 
enabling a new iteration of the analysis.

Row 0 in~\cref{fig:sleec} shows the performance of this custom algorithm 
over a large set of \sleec benchmarks~\cite{zenodo}.
We consider it as the baseline for comparison with 
the quantified approaches described next.
%
%
\tiny
\begin{table}[t]
  \begin{center}
    \caption{\sleec performance (20s \timeout).}
    \label{fig:sleec}
    \begin{tabular}{lrrr}
      \toprule
      Technique                            & \sat          & \unsat       &  \unknown \\
      \midrule
      0 LEGOS (z3)     & 1204          & 158          & 0                      \\
      \midrule
      1 \(\forall / \exists\) pred (z3)     & 282           & \textbf{158} & 922                     \\
      2 \(\forall / \exists\) pred-mbqi (\cvc)   & 249           & 73           & 1040                    \\
      3  \(\forall / \exists\) sets-enum (\cvc)   & 0             & 147          & 1215                    \\
      4  \(\forall / \exists\) sets-fmf  (\cvc)   & 486           & 1            & 875                     \\
      5 set.all/set.some                 (\cvc)  & \textbf{1189} & 138          & \textbf{35}                      \\
      \bottomrule
    \end{tabular}
  \end{center}
\end{table}
\normalsize
The authors of LEGOS aimed to simplify their
FOL* satisfiability checking algorithm for
simplicity and maintainability by directly
using quantifiers in first-order logic (FOL).
Fortunately, FOL* formulas can be translated
into FOL by mapping relational classes to
either sets or uninterpreted predicates.
Specifically, a class $C$ of relational objects with $n$
attributes in FOL* can be mapped to a set $s_C$
of $n$-arity tuples
or to an n-arity predicate $p_C$.
Quantifiers over relational objects of class $C$
can be converted into quantifiers over either
elements of $s_C$ or tuples that satisfy $p_C$.

The remaining five rows in~\cref{fig:sleec} show the results of
different FOL encodings with quantifiers.
Rows 1 and 2 encode sets as predicates and show the performance in z3 
and \cvc. 
By default, z3 enables model-based quantifier instantiation technique (mbqi)
and we enabled it for \cvc~\cite{mbqi} as well.
z3 managed to solve all \unsat benchmarks, and was overall superior to \cvc,
but it struggled with \sat benchmarks.
Rows 3 and 4 correspond to an encoding in the theory of finite sets and relation 
of \cvc, along with standard quantifiers.
Row 3 uses enumerative instantiation~\cite{enumerativeInstantiation} 
to solve more \unsat benchmarks,
whereas finite model finding~\cite{fmf,findModelFind} is used to solve 
more \sat benchmarks in Row 4.
The last row shows the performance of set-bounded quantifiers
using the filter operator \(\filter\) as explained in~\cref{sec:all_some}.
This approach is the clear winner on \sat benchmarks among
the quantified approaches and is very competitive on \unsat benchmarks.
Our approach failed to solve 35 \unsat benchmarks, and 
upon investigation, we found that the solver keeps generating 
new element terms because of the presence of nested and alternating 
\(\setAll\) and \(\setSome\) quantifiers. 
Although this falls outside the decidable fragment, 
this limitation occurs in realistic benchmarks, so we plan to address it
in future research. 

In general, the results imply that using set-bounded quantifiers is more effective
than using standard quantifiers in \sat benchmarks. 
Figure~\ref{fig:cactus} shows the number of benchmarks solved by each technique.
All experiments were run on a Linux machine with 16-Core AMD EPYC 7313 
processor. 
We used z3 version v4.13.4 and \cvc version 1.2.0 (commit f3fc80e). 
All benchmarks used in this paper are publicly available.
Although our approach could not solve 35 benchmarks compared to LEGOS,
it took less time to solve the rest compared to LEGOS.
Our intuition here is that LEGOS takes more time because 
it may need multiple abstraction refinement rounds and so multiple calls
to z3 before coming to a conclusion
whereas our approach requires only one call to \cvc. 
It is worth mentioning that developing the encodings to \cvc for \sleec
revealed a few soundness bugs in \sleec's FOL* encoding.
For example, one bug caused by an improper handling of name
collisions in quantified formulas was
identified through discrepancies in
satisfiability results compared to \cvc.
%
This led to multiple rounds of bug fixes in both \sleec and \cvc.
\begin{figure}
  \begin{center}
    \includegraphics[scale=0.45]{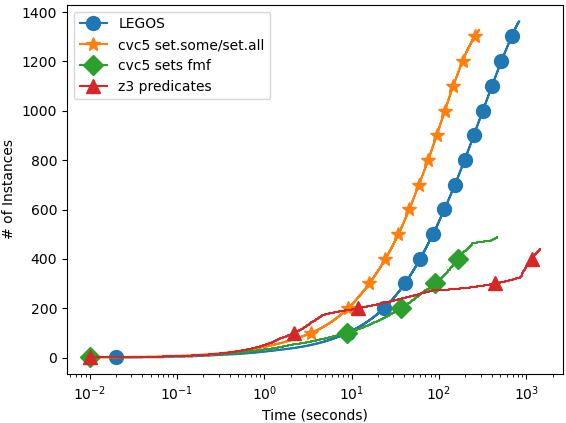}
  \end{center}
  \caption{Cactus plot showing the number of instances solved
    by each technique.}
  \label{fig:cactus}
  \vspace{-3ex}
\end{figure}
%




%
%
\section{Conclusion and Future Work}
We proved that the satisfiability of qffs in the theory of finite relations 
with filter and cross product operators is decidable
when the satisfiability of qffs in the element theory is decidable and 
no set terms are used in filter predicates. 
We also proved that allowing set terms in filter predicates does compromise
decidability.
We demonstrated how to express set-bounded quantifiers using filter terms,
and showed that they are superior to standard quantifiers for SMT solving
with a real set of benchmarks, especially with satisfiable formulas. 
%

%

Our calculus is incomplete 
when we add the cardinality operator 
along with its rules in~\cite{sets}.
%
%
We leave addressing this problem to future work. 

Multiset-bounded quantifiers were not discussed in this paper.
However, similarly to set quantifiers, they could be easily implemented
using the filter operator for multisets in~\cite{lparBags}.
%
This is also left to future work.

\bibliographystyle{splncs04}
\bibliography{main.bib}
\onecolumn
\appendix
\subsection{Proof of \cref{lem:setsConstraints}}
\setsConstraints*
\begin{proof}
  The formula \(\varphi\) can be transformed into an equisatisfiable disjunctive normal form
  \(\varphi_1 \vee \dots \vee \varphi_n\) using standard techniques,
  where \(\varphi_i\) is a conjunction of \(\varphi_i^1, \dots, \varphi_i^{k_i}\)
  of literals for \(i \in [1,n]\).
  Each literal is either a relation constraint
  or an element constraint.
  For each \(i \in [1,n]\) define:
  \begin{align*}
    \Sc_i & = \{\varphi_i^j \ver \varphi_i^j \text{ is a relation constraint}\}       \\    
    \Ec_i & = \{\varphi_i^j \ver \varphi_i^j \text{ is an element constraint}\}
  \end{align*}
  where \(j \in [1, k_i]\) for each \(i \in [1,n]\).
  It is clear that \(\varphi\) is satisfiable if and only if
  \(\Sc_i \cup \Ec_i\) is satisfiable for some \(i\in [1,n]\).
\end{proof}
\subsection{Proof of \cref{lem:setFilterSound}}
\setFilterSound*
\begin{proof}
  The proof follows from the semantics of set operators and the definition of \(\Sc^*\) and \(\Ec^*\).
  We show the proof for rules \ruleSetFilterDown and \ruleSetFilterUp. The other cases are similar. 
  Suppose the premise configuration is \(c_1 = (\Sc_1, \Ec_1)\). 
  Rule \ruleSetFilterDown has only one conclusion, say 
  \(c_2 = (\Sc_1 \cup \{e \sqin s\}, \Ec_1 \cup p(e))\). 
  For this case, we prove a stronger claim that 
  any model satisfies \(c_1\) iff it satisfies \(c_2\). 
  Suppose \(\I_1\) is a model for \(c_1\). 
  We prove that \(\I_1\) is also a model for \(c_2\). 
  We need to show that \(\I_1\) satisfies the new constraints \(e \sqin s\) and \(p(e)\) 
  since it already satisfies \(\Sc_1\) and \(\Ec_1\) from our assumption. 
  The rule is applied because \(e' \sqin \filter(p, s') \in \Sc_1\) for some \(e', s' \in \ter{\Sc_1}\)
 such that \(e \teq e'\) and \(s \teq s'\) in \(\Sc_1\).  
 Therefore,  \(\I_1\) satisfies \(e \sqin \filter(p, s)\) and 
 \(\I_1(e) \in \I_1(\filter(p, s))\). 
 From the semantics of \(\filter(p, s)\), 
 it must be the case that \(\I(p(e))\) holds and \(\I(e) \in \I(s)\).
 Therefore \(\I_1\) satisfies \(e \sqin s\) and \(p(e)\), and hence satisfies \(c_2\).

 For other direction, suppose \(\I_2\) satisfies \(c_2 \). 
 Since \(\Sc_1 \subseteq \Sc_1 \cup \{e \sqin s\}\) 
 and \(\Ec_1 \subseteq \Ec_1 \cup p(e)\), 
 then \(\I_2\) satisfies \(c_1 = (\Sc_1, \Ec_1)\). 
 Therefore, rule \ruleSetFilterDown is sound. 

  For rule \ruleSetFilterUp, suppose the premise configuration \(c_1 = (\Sc_1, \Ec_1)\) 
  is satisfied by a model \(\I_1\).
  In that model, we have either \(\I(p(e))\) holds or \(\neg \I(p(e))\) holds.
  If \(\I(p(e))\) holds, then \(\I_1\) is a model for the first branch of the conclusion, 
  because \(\I(e) \in \I(\filter(p,s))\).
  If \(\neg\I(p(e))\) holds, then \(\I_1\) is a model for the second branch, 
  since  \(\I(e) \not\in \I(\filter(p,s))\). 
  Now suppose \(\I_2\) is a model for one of the two possible configurations   
  \(c_2 = (\Sc_1 \cup \{e \sqin \filter(p,s)\}, \Ec_1 \cup p(e))\)
  or \(c'_2 = (\Sc_1 \cup \{e \not\sqin \filter(p,s)\}, \Ec_1 \cup \neg p(e))\)
  in the conclusion. 
  Since \(e \sqin s \in \Sc_1^*\) in \(c_2\) and \(c'_2\), 
  then \(\I_2\) satisfies \(e \sqin s\). 
  Therefore \(\I_2\) is a model for the premise configuration.
\end{proof}

\subsection{Completeness Proof}\label{sec:completenessProof}
Here we prove the remaining cases for~\cref{prop:setModel}. 
  Recall that we assume, with no loss of generality, that all set equalities in
  \(\Sc_\mathrm{l}\) have the form \(x \teq t\) where \(x\) is a variable 
  that never occurs in the right-hand side of an equality of \(\Sc_\mathrm{l}\).
  This restriction allows us to define a well-founded ordering 
  \(\succ_\mathrm{T}\) over the set variables in \(\Sc_\mathrm{l}\)
  where \(x \succ_\mathrm{T} y\) iff \(x \teq t \in \Sc_\mathrm{l}\) and
  \(y\) occurs in \(t\).
  We argue by induction on \(\succ_\mathrm{T}\).
\begin{enumerate}
  \item[3.] \(s \teq [\empty]\). Rule \ruleEmptyUnsat would apply to the configuration if there was a constraint of
  the form \(x \sqin s \in \Sc^*\). 
  Since there is none, it follows that \(\I(s) = \{\} = \I([])\).
\item[4.] \(s \teq [x]\). It is sufficient to show that \(\I(s) = \{\I(x)\}\).
  Since rule \ruleSingleUp   is not applicable,
  we can conclude that \(x \sqin s \in \Sc^*\).
  It follows that \(\{\I(x)\} \subseteq \I(s)\).
  The other direction, \(\I(s) \subseteq \{\I(x)\} \), follows because of
  saturation with respect to rule \ruleSingleDown.
  \item[5.] \(s \teq t \sqcap u\).
  We need to show that \(\I(s) = \I(t) \cap \I(u) \).
  The proof of the left-to-right inclusion depends on rule \ruleInterDown:
  \begin{align*}
    x & \in \I(s)                                                                                                                      \\
    x & = \I(e) \text{ for some } e \text{ with } e \sqin s \in \Sc^* \tag{definition of \(\I(s)\)}                                    \\
    x & = \I(e) \text{ for some } e \text{ with } e \sqin t \sqcap u \in \Sc^* \tag{definition of \(\Sc^*\) and \(s \teq t \sqcap u\)} \\
    x & = \I(e),  e \sqin t \in \Sc, e \sqin u \in \Sc \tag{rule \ruleInterDown}                                                       \\
    x & \in \I(t) \text{ and }  x \in \I(u) \tag{by induction hypothesis since \(s \succ_\mathrm{T} t\) and \(s \succ_\mathrm{T} u\)}                                                                            \\
    x & \in \I(t) \cap \I(u)
  \end{align*}
  For the other direction, \(\I(t) \cap \I(u) \subseteq \I(s) \), we rely on rule \ruleInterUp:
  \begin{align*}
    x & \in \I(t) \cap \I(u)                                                                                                       \\
    x & \in \I(t) \text{ and } x \in \I(u)                                                                                         \\
    x & = \I(e_1) = \I(e_2) \text{ for some } e_1, e_2 
    \\ & \qquad \text{ with } e_1 \sqin t \in \Sc^*, e_2 \sqin u \in \Sc^* 
    \tag{by induction hypothesis since \(s \succ_\mathrm{T} t\) and \(s \succ_\mathrm{T} u\)}
    \\            
    x & = \I(e_1) = \I(e_2), e_1 \sqin t \in \Sc^*, e_2 \sqin u \in \Sc^*, e_1 \teq e_2 \in \Sc^* 
    \tag{saturation of rule \ruleEIdent}  \\
    x & = \I(e_1), e_1 \sqin t \in \Sc^*, e_1 \sqin u \in \Sc^*                                                                    \\
    x & = \I(e_1), e_1 \sqin t \sqcap u \in \Sc^* \tag{rule \ruleInterUp}                                                          \\
    x & = \I(e_1), e_1 \sqin s \in \Sc^* \tag{\(s \teq t \sqcap u\)}                                                               \\
    x & \in \I(s) \tag{definition of \(\I(s)\)}
  \end{align*}
  Notice that if we were to choose the second conclusion in rule \ruleEIdent, 
  \(e_1 \tneq e_2 \in \Ec\) would contradict that \(\I(e_1) = \I(e_2)\). 
  \item[6.] \(s \teq t \sqcup u\).
        We need to show that \(\I(s) = \I(t) \cup \I(u) \).
        The proof of the left-to-right inclusion depends on rule \ruleInterDown:
        \begin{align*}
          x & \in \I(s)                                                                                                                      \\
          x & = \I(e) \text{ for some } e \text{ with } e \sqin s \in \Sc^* \tag{definition of \(\I(s)\)}                                    \\
          x & = \I(e) \text{ for some } e \text{ with } e \sqin t \sqcup u \in \Sc^* \tag{definition of \(\Sc^*\) and \(s \teq t \sqcup u\)} \\
          x & = \I(e),  e \sqin t \in \Sc \text{ or } e \sqin u \in \Sc \tag{rule \ruleUnionDown}                                            \\
          x & \in \I(t) \text{ or }  x \in \I(u) \tag{by induction hypothesis since \(s \succ_\mathrm{T} t\) and \(s \succ_\mathrm{T} u\)}                                                                           \\
          x & \in \I(t) \cup \I(u)
        \end{align*}
        For the other direction, \(\I(t) \cup \I(u) \subseteq \I(s) \):
        \begin{align*}
          x & \in \I(t) \cup \I(u)                                                                                            \\
          x & \in \I(t) \text{ or } x \in \I(u)                                                                               \\
          x & = \I(e)  \text{ for some } e \text{ with } e \sqin t \in \Sc^* \text { or } e \sqin u \in \Sc^* 
          \tag{by induction hypothesis since \(s \succ_\mathrm{T} t\) and \(s \succ_\mathrm{T} u\)} \\
          x & = \I(e), e \sqin t \sqcup u \in \Sc^* \tag{\(t \sqcup u \in \ter{\Sc}\), rule  \ruleUnionUp}                    \\
          x & = \I(e), e \sqin s \in \Sc^* \tag{\(s \teq t \sqcup u\)}                                                        \\
          x & \in \I(s) \tag{definition of \(\I(s)\)}
        \end{align*}
  \item[7.] \(s \teq t \setminus u\).
        We need to show that \(\I(s) = \I(t) \setminus \I(u) \).
        The proof of the left-to-right inclusion depends on rule \ruleDifferenceDown:
        \begin{align*}
          x & \in \I(s)                                                                                                                            \\
          x & = \I(e) \text{ for some } e \text{ with } e \sqin s \in \Sc^* \tag{definition of \(\I(s)\)}                                          \\
          x & = \I(e) \text{ for some } e \text{ with } e \sqin t \setminus u \in \Sc^* \tag{definition of \(\Sc^*\) and \(s \teq t \setminus u\)} \\
          x & = \I(e),  e \sqin t \in \Sc^* \text{ and } e \not\sqin u \in \Sc^* \tag{rule \ruleDifferenceDown}                                    \\
          x & \in \I(t) \text{ and } x \not\in \I(u)           
          \tag{by induction hypothesis since \(s \succ_\mathrm{T} t\) and contradiction below}                                         \\
          x & \in \I(t) \setminus \I(u)
        \end{align*}
        Suppose by contradiction that \(x \in \I(u)\). Since 
        \(s \succ_\mathrm{T} u\)        
        ,by induction hypothesis there exists \(e_2\) such that \(x = \I(e) = \I(e_2)\) and
        \(e_2 \sqin u \in \Sc^*\).
        By saturation of \ruleEIdent either \(e \teq e_2 \in \Sc^* \cap \Ec^*, \) or
        \( e \tneq e_2 \in \Sc^* \cap \Ec^*\).
        The latter would lead to a contradiction because \(\I(e) = \I(e_2)\).
        The former would imply \(e \sqin u  \in \Sc^*\) which would trigger rule
        \ruleSetUnsat along with \(e \not\sqin u  \in \Sc^*\) to make the current branch \(\unsat\).
        This contradicts our assumption that the current branch is open.

        For the other direction, \(\I(t) \setminus \I(u) \subseteq \I(s) \):
        \begin{align*}
          x & \in \I(t) \setminus \I(u)                                                                                               \\
          x & \in \I(t) \text{ and } x \not\in \I(u)                                                                                  \\
          x & = \I(e)  \text{ for some } e \text{ with } e \sqin t \in \Sc^* 
          \tag{by induction hypothesis since \(s \succ_\mathrm{T} t\)}\\
          x & = \I(e), e \sqin t \setminus u \in \Sc^* \tag{\(t \setminus u \in \ter{\Sc}\), rule  \ruleDifferenceUp discussed below} \\
          x & = \I(e), e \sqin s \in \Sc^* \tag{\(s \teq t \setminus u\)}                                                             \\
          x & \in \I(s) \tag{definition of \(\I(s)\)}
        \end{align*}
        Rule \ruleDifferenceUp splits into two branches.
        One with \(e \sqin u \in \Sc^*\) and the other is the one we mentioned above \(e \sqin t \setminus u \in \Sc^*\).
        Suppose by contradiction that the current branch has \(e \sqin u \in \Sc^*\).
        Since \(s \succ_\mathrm{T} u\), by our induction hypothesis we have \(x \in \I(u)\) which contradicts \(x \not\in \I(u)\).
  \item[8.] \(s \teq t \product u\).
        We need to show that \(\I(s) = \I(t) \times \I(u) \).
        The proof of the left-to-right inclusion depends on rule \ruleProductDown:
        \begin{align*}
          \tup{x_1, \ldots, x_m, y_1, \ldots, y_n} & \in \I(s)                                                                                                         \\
          \tup{x_1, \ldots, x_m, y_1, \ldots, y_n} & = \tup{\I(a_1), \ldots, \I(a_m), \I(b_1), \ldots, \I(b_n)}                                                        \\
                                                   & \quad \text{ for some }  \tup{a_1, \ldots, a_m, b_1, \ldots, b_n} \sqin s \in \Sc^* \tag{definition of \(\I(s)\)} \\
          \tup{a_1, \ldots, a_m, b_1, \ldots, b_n} & \sqin t \product u \in \Sc^* \tag{definition of \(\Sc^*\) and \(s \teq t \product u\)}                            \\
          \tup{a_1, \ldots, a_m}                   & \sqin t \in \Sc, \tup{b_1, \ldots, b_n} \sqin u \in \Sc \tag{rule \ruleProductDown}                               \\
          \tup{x_1, \ldots, x_m}                   & \in \I(t), \tup{y_1, \ldots, y_n} \in \I(u) 
         \tag{by induction hypothesis since \(s \succ_\mathrm{T} t\) and \(s \succ_\mathrm{T} u\)} 
                                                                \\
          \tup{x_1, \ldots, x_m, y_1, \ldots, y_n} & \in \I(t) \times \I(u).
        \end{align*}
        For the other direction, \(\I(t) \times \I(u) \subseteq \I(s) \):
        \begin{align*}
          \tup{x_1, \ldots, x_m, y_1, \ldots, y_n}                 & \in \I(t) \times \I(u) \text{ implies }                                                                            \\
          \tup{x_1, \ldots, x_m}                                   & \in \I(t), \tup{y_1, \ldots, y_n} \in \I(u)                                                                        \\
          \tup{x_1, \ldots, x_m}                                   & = \tup{\I(a_1), \ldots, \I(a_m)}, \tup{y_1, \ldots, y_n} =  \tup{\I(b_1), \ldots, \I(b_n)}                         \\
                                                                   & \quad  \text{ for some }  \tup{a_1, \ldots, a_m} \sqin t, \tup{b_1, \ldots, b_n} \sqin u \in \Sc^* 
                                                                   \tag{by induction hypothesis since \(s \succ_\mathrm{T} t\) and \(s \succ_\mathrm{T} u\)} 
                                                                   \\
          \tup{a_1, \ldots, a_m, b_1, \ldots, b_n}                 & \sqin t \product u \in \Sc^* \tag{\(t \product u \in \ter{\Sc}\), rule  \ruleProductUp}                            \\
          \tup{a_1, \ldots, a_m, b_1, \ldots, b_n}                 & \sqin s \in \Sc^* \tag{\(s \teq t \product u\)}                                                                    \\
          \tup{\I(a_1), \ldots, \I(a_m), \I(b_1), \ldots, \I(b_n)} & \in \I(s) \tag{definition of \(\I(s)\)}                                                                            \\
          \tup{x_1, \ldots, x_m, y_1, \ldots, y_n}                 & \in \I(s) \tag{Property~\ref{prop:one}}.
        \end{align*}
\end{enumerate}
\subsection{\(\sigmap\) grammar}
Here we provide the grammar for language \(\sigmap\) used in 
Section~\ref{sec:undecidable_set_fragments}.
\begin{figure*}[h]
  \begin{center}
    \begin{bnf}
      \(\varphi\) : \textsf{Set Formula} ::=
      \(e \sqin s\) //
      \(s_1 \sqsubseteq s_2\) //
      \(e_1 \teq e_2\) //
      \(s_1 \teq s_2\)
      | \(\varphi_1 \wedge \varphi_2\) //
      \(\varphi_1 \vee \varphi_2\) //
      \(\neg\varphi_1\);;
      \(e\) : \textsf{Element term} ::=
      \(k\)
      // \(x\)
      // \( \tup{x_1, x_2}    \);;
      \(s\) : \textsf{Set term} ::=
      \(x\)
      // [e]
      // \(s_1 \sqcup s_2\)
      // \(s_1 \sqcap s_2\)
      // \(\map(\lambdaExpr, s_1)\)
      ;;
      \(\lambdaExpr\) : \textsf{Lambda term}  ::=
      \(\lambda x. \arithExpr\)
      // \(\lambda \tup{x_1, x_2}. \arithExpr\)
      // \(\lambda \tup{x_1, x_2}. \tup{\arithExpr_1, \arithExpr_2}\)
      ;;
      \(\arithExpr\) : \textsf{LIA term}  ::=
      k
      // \(x\)
      // \(k \cdot x\)
      // \(\arithExpr_1 + \arithExpr_2\)
      // \(- \arithExpr_1\)
      | \(\ite(\varphi^{\lian},\arithExpr_1, \arithExpr_2)\)
      ;;
      \(\varphi^{\lian}\) : \textsf{LIA constraint} ::=
      \(\arithExpr_1 \teq \arithExpr_2\)
      // \(\arithExpr_1 > \arithExpr_2\)
      // \(\arithExpr_1 \geq \arithExpr_2\)
      | \(\varphi^{\lian}_1 \wedge \varphi^{\lian}_2\)
      // \(\varphi^{\lian}_1 \vee \varphi^{\lian}_2\)
      // \(\neg \varphi^{\lian}_1 \)
      ;;
    \end{bnf}
  \end{center}
  \caption{\(\sigmap\) grammar.}
  \label{fig:grammar}
\end{figure*}
\end{document}